\documentclass[journal]{IEEEtran}
\usepackage{pifont}
\usepackage{times,amsmath,color,amssymb,graphicx,epsfig,cite,psfrag,subfigure}
\usepackage{mathrsfs}
\usepackage{amsfonts}
\usepackage{algorithm}
\usepackage{algorithmic}
\long\def\symbolfootnote[#1]#2{\begingroup%
\def\thefootnote{\fnsymbol{footnote}}\footnote[#1]{#2}\endgroup}

\newtheorem{pro}{\rm{\underline{\textbf{Problem}}}}

\newtheorem{prop}{\rm{\underline{\textbf{Proposition}}}}

\newtheorem{thm}{\rm{\underline{\textbf{Theorem}}}}
\newtheorem{lem}{\rm{\underline{\textbf{Lemma}}}}

\begin{document}
\title{Cost minimization for fading channels with energy harvesting and conventional energy}
\author{~~~~Xin~Kang,~\IEEEmembership{Member,~IEEE}, ~Yeow-Khiang
Chia,~\IEEEmembership{Member,~IEEE}, ~Chin Keong
Ho,~\IEEEmembership{Member,~IEEE}, ~and \newline ~~Sumei
Sun,~\IEEEmembership{Senior Member,~IEEE}
\thanks{X. Kang, Y. -K. Chia, , C. K. Ho and S. Sun are with Institute for Infocomm Research, 1 Fusionopolis Way,
$\#$21-01 Connexis, South Tower, Singapore 138632 (E-mail: \{xkang,
chiayk, hock, sunsm\}@i2r.a-star.edu.sg).} } \markboth{X. Kang et.
al, ``Cost minimization for fading channels with energy harvesting
and conventional energy''} {} \maketitle

\begin{abstract}
In this paper, we investigate resource allocation strategies for a
point-to-point wireless communications system with hybrid energy
sources consisting of an energy harvester and a conventional energy
source. In particular, as an incentive to promote the use of
renewable energy, we assume that the renewable energy has a lower
cost than the conventional energy. Then, by assuming that the
non-causal information of the energy arrivals and the channel power
gains are available, we minimize the total energy cost of such a
system over $N$ fading slots under a proposed outage constraint
together with the energy harvesting constraints. The outage
constraint requires a minimum fixed number of slots to be reliably
decoded, and thus leads to a mixed-integer programming formulation
for the optimization problem. This constraint is useful, for
example, if an outer code is used to recover all the data bits.
Optimal linear time algorithms are obtained for two extreme cases,
i.e., the number of outage slot is $1$ or $N-1$. For the general
case, a lower bound based on the linear programming relaxation, and
two suboptimal algorithms are proposed. It is shown that the
proposed suboptimal algorithms exhibit only a small gap from the
lower bound. We then extend the proposed algorithms to the
multi-cycle scenario in which the outage constraint is imposed for
each cycle separately. Finally, we investigate the resource
allocation strategies when only causal information on the energy
arrivals and only channel statistics is available. It is shown that
the greedy energy allocation is optimal for this scenario.
\end{abstract}


\section{Introduction}
Driven by environmental concerns, green wireless communications have
recently attracted increasing attention from both industry and
academia. It is reported in \cite{ICTtrends} that the world-wide
cellular networks consume about sixty billion kilowatt hour (kWh) of
energy per year, which result in a few hundred million tons of
carbon dioxide emission yearly. These figures are expected to
increase rapidly in the near future if no further actions are taken.
On the other hand, it is pointed out in \cite{LiYe} that, with the
explosive growth of high data rate wireless applications, more
energy is consumed to guarantee the users' quality of service (QoS).
These facts create a compelling need for green wireless
communications. One way to achieve green wireless communications is
to improve the energy-efficiency of the current communications
networks \cite{chen2011fundamental}. Another way is to introduce
clean and environment-friendly renewable energy (such as solar power
and wind power) to wireless communications networks
\cite{barton2004energy}.


Introducing energy harvesting capabilities to wireless
communications is a promising approach to achieve green
communications, with its great potential to reduce the carbon
dioxide emission produced by conventional energy.  However, it poses
lots of new challenges on the design of resource allocation
strategies for the wireless communications networks. This is mainly
due to the highly time-varying availability of the renewable energy.
For instance, solar energy and wind energy may vary significantly
over time and locations depending on the weather and the climate
conditions. Thus, conventional transmit power constraints are not
suitable to model communications devices with renewable energy.
Instead, resource has to be allocated subject to energy harvesting
constraints. With energy harvesting constraints, in every time slot,
the transmitter is allowed to use at most the amount of harvested
and stored energy currently available. In other words, the
transmitter can not consume any energy harvested in future.

Throughput optimization for wireless communications systems with
such energy harvesting constraints has been extensively studied in
recent literatures. The capacity of AWGN channel with the energy
harvesting system setup was studied in \cite{ozel2010information}
and \cite{rajesh2011capacity}. Throughput maximization for a
single-user energy harvesting system with a deadline constraint in a
static channel was studied in \cite{yang2012optimal} and
\cite{tutuncuoglu2012optimum}. For single-user fading channel, the
optimal energy allocation scheme to maximize the throughput for a
slotted system over a finite horizon of time slots was obtained in
\cite{ho2012optimal} through dynamic programming. In
\cite{ozel2011transmission}, the authors derived continuous time
optimal policies to maximize the throughput of fading channels with
the energy harvesting constraints. Then, energy allocation
strategies to maximize the throughout of multiple access channels
and broadcast channels with energy harvesting constraints were
investigated in \cite{yang2012optimalMAC} and
\cite{yang2012broadcasting}. Throughput maximization for relay
channels with energy harvesting constraints was studied in
\cite{huang2011throughput}.

Another line of related research in wireless communications with
energy harvesting nodes focused on simultaneous wireless information
and power transfer \cite{varshney2008transporting,grover2010shannon,
zhang2013mimo, liu2013wireless,luo2013optimal, zhou2012wireless,
XKangFD2014}. The idea of simultaneous wireless information and
power transfer was proposed in \cite{varshney2008transporting}. 
In \cite{grover2010shannon}, the authors studied the tradeoff
between information rate and power transfer in a frequency selective
wireless system. Then, the tradeoff between energy and information
for a MIMO broadcast system was studied in \cite{zhang2013mimo}. In
\cite{liu2013wireless,luo2013optimal}, operation protocols and
switching schemes to minimize the outage probability for wireless
systems with simultaneous information and power transfer were
studied. In \cite{zhou2012wireless}, practical receiver designs for
implementing simultaneous information and power transfer were
investigated. In \cite{XKangFD2014}, a new protocol was proposed to
achieve simultaneous bi-directional wireless information and power
for a multi-user communication network.

In these aforementioned works, the communications devices are
powered only by the renewable energy. However, due to the highly
random availability of the renewable energy, communications devices
powered only by the renewable energy may not be able to guarantee a
required level of QoS. Since in many communication systems, such as
in a cellular communications, the QoS must be satisfied at least
with high probability, a hybrid energy supply system with both
renewable energy supply and conventional energy supply is preferred
in practice. This motivated us to consider a communications system
with both energy harvesters and conventional energy supply in this
paper. As an incentive to promote the use of  the renewable energy,
we assume that the renewable energy has a lower cost than the
conventional energy. Under this assumption, minimizing the energy
consumption is not equivalent to minimizing the energy cost. In this
paper,  unlike the conventional energy-efficient studies whose
objective is minimizing the energy consumption, our objective is to
minimize the total energy cost. The motivation for this is that,
from a user's perspective, minimizing the total energy cost is more
important and meaningful. \textcolor[rgb]{0.00,0.00,0.00}{It is
worth pointing out that hybrid energy supply model was also
considered in \cite{chia2013energy,ng2013energy,OOISIT,ng2013Dec}.
However, the focus of \cite{chia2013energy} was to develop the
energy cooperation scheme between two cellular base stations. The
target of \cite{ng2013energy} is to derive the resource allocation
scheme to maximize the weighted energy efficiency of data
transmission over a downlink orthogonal frequency division multiple
access (OFDMA) system. The objective of \cite{OOISIT} and
\cite{ng2013Dec} were to maximize the throughput and minimize the
energy consumption for a point-to-point channel, respectively.}


The contribution and the main results of this paper are summarized
as follows.
\begin{itemize}
  \item  We consider a point-to-point communications system with both renewable energy and conventional energy. To guarantee the QoS of the system, we propose an outage constraint, which requires a minimum number of slots to be reliably decoded.
  This constraint is useful if, for example, an outer code is used to recover all data bits. A mixed integer programming problem is then formulated to minimize the total energy cost of such a system for $N$ fading
  slots under both energy harvesting constraints and the
  proposed outage constraint.
  \item  We study the optimal power allocation
  strategy to minimize the total energy cost by assuming full knowledge of the channel and energy state
  information (CESI). By exploring the structured properties of the optimal solution, we propose two
  low complexity algorithms with worst case linear time complexity to yield the
  optimal power allocation for two extreme cases: when the number of outage slot is either $1$ or
  $N-1$.
  \item For the general case, we propose a lower bound based on linear programming relaxation. Besides, we propose two
  suboptimal algorithms referred to as \emph{linear programming based channel removal} (LPCR) and \emph{worst channel removal} (WCR), respectively.
  It is shown by simulation that the proposed suboptimal algorithms
exhibit only a small gap with respect to the lower bound. It is
proved that WCR is optimal when certain
  conditions are satisfied.
  \item We extend the proposed algorithms and the obtained
  results to the multi-cycle scenario where the outage constraint is imposed for each cycle separately. It is shown that the proposed
  algorithms can be easily extended to the multi-cycle scenario with
  few modifications.
   \item When CESI is not
  available, a new outage constraint is proposed. Closed-form solution is
  obtained for this case. It is shown that the optimal solution
  has a greedy feature. It always uses the low-cost energy first and
  uses the high-cost energy only when necessary (i.e., when the low-cost energy is not enough to guarantee the required QoS).
\end{itemize}

The rest of this paper is organized as follows. We describe the
system model in Section \ref{Sec-SysModel} and give out the problem
formulation in Section \ref{Sec-ProblemFormulation}. The proposed
algorithms and obtained theoretical results are then presented in
Section \ref{Sec-TheoResults}. In Section \ref{Sec-NoCSI}, we
investigate the power allocation strategy when future CSI and energy
harvesting information is not available. Then, in Section
\ref{Sec-NumericalResults}, numerical results are presented to
verify the proposed studies. Finally, Section \ref{Sec-Conclusions}
concludes the paper.

\section{System Model}\label{Sec-SysModel}
In this paper, we consider a point-to-point channel with one
transmitter (Tx) and one receiver (Rx). We assume that the
transmitter has access to two types of energy: conventional energy
and renewable energy. The conventional energy is obtained from
conventional power grid or batteries. The renewable energy is
obtained by energy harvesting devices, such as solar panel and wind
turbine. These two types of energy are provided to the transmitter
at different prices: $\alpha$ per unit for conventional energy, and
$\beta$ per unit for renewable energy. For exposure, we make the
following assumptions in this paper.
\begin{figure}[t]
        \centering
        \includegraphics*[width=8.5cm]{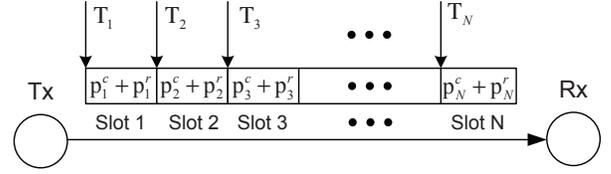}
        \caption{System Model}
        \label{Fig-SystemModel}
\end{figure}

\begin{itemize}
      \item $\alpha>\beta$. We make this assumption due to the following
two reasons: \textcolor[rgb]{0,0,0}{(i)} The renewable energy
greatly depends on the environment (such as the weather), and thus
is not as reliable as the conventional energy. Therefore, the
renewable energy should be priced lower to attract users.
\textcolor[rgb]{0,0,0}{(ii)} The renewable energy is clean and
environment-friendly. Thus, pricing the renewable energy at a lower
price provides an incentive for users to use green energy.
      \item  The transmission is slotted, and the Tx is equipped with an
      energy storage device. The energy harvested at the beginning of slot $i$ is
      denoted by $T_i$.  Thus, the harvested energy $T_i$  can be
      used is slot $i$ or
      stored for future use. The conventional energy and the renewable energy
      consumed at slot $i$ are denoted as $p_i^c$ and $p_i^r$,
      respectively.

      \item  The channel experiences block-fading and remains
constant during each transmission slot, but possibly changes from
one slot to another. The channel power gain is assumed to be a
random variable with a continuous probability density function (PDF)
$f(x)>0, \forall x>0$.  The channel power gain for slot $i$ is
denoted as $g_i$. The noise at the Rx is assumed to be a circular
symmetric complex Gaussian random variable with zero mean and
variance $N_0$ denoted by $\mathcal {CN}(0, N_0)$.
      \end{itemize}

\section{Problem Formulation}\label{Sec-ProblemFormulation}
In this paper, we assume that the whole transmission process
consists of $N$ time slots. Under the system model given in Section
\ref{Sec-SysModel},  the instantaneous transmission rate in slot $i$
can be written as $\ln(1+\frac{g_i\left(p_i^c +
p_i^r\right)}{N_0})$. Thus, if the target transmission rate of the
user is $R$, the minimum power required to support this rate is
\begin{align}p_i^{inv}=\frac{N_0~\left(e^R-~1\right)}{g_i}.\end{align}

We refer to $p_i^{inv}$ as \textbf{\emph{channel inversion power}}
for slot $i$. If $p_i^c + p_i^r< p_i^{inv}$, we say the user is in
outage in slot $i$. For convenience, we define an indicator function
for each slot, which is given as
\begin{align}\label{Eq-Chi}
\chi_i\left(p_i^c,p_i^r\right)=\left\{\begin{array}{ll}
                1, & ~\mbox{if}~\ln\left(1+\frac{g_i\left(p_i^c +
p_i^r\right)}{N_0}\right)<R, \\
                0, & ~\mbox{otherwise}.
              \end{array}
\right.
\end{align}

To guarantee the QoS, we assume that the fraction of outage should
be kept below a prescribed target $\epsilon$. Mathematically, this
can be written as
\begin{align}
\frac{1}{N}\sum_{i=1}^N\chi_i\left(p_i^c,p_i^r\right) \le \epsilon,
\end{align}
where $\chi_i\left(p_i^c,p_i^r\right)$ is given by \eqref{Eq-Chi}.
In this paper, we refer to this constraint as \textbf{\emph{outage
constraint}}. This outage constraint requires at least $\lceil
N(1-\epsilon) \rceil$ packets to be received without error over $N$
slots, which is useful for delay-sensitive data or when an outer
code is used that can correct any $\lfloor N(1-\epsilon) \rfloor$
packets in outage. Clearly, if $\epsilon=0$, no outage is allowed.

In this paper, we assume that the harvested energy can be stored for
future use. Thus, the \emph{energy harvesting constraints} can be
written as
\begin{align}
\sum_{i=1}^k p_i^r- \sum_{i=1}^k T_i \le 0, \forall k \in
\left\{1,2,\cdots,N\right\}.
\end{align}

In this paper, our objective is to minimize the total energy cost of
the $N$-slot transmission through proper energy allocation
strategies. Under the constraints described above, the problem can
be formulated as follows.
\begin{pro}\label{Problem-3}
\begin{align} \min_{p_i^c,~p_i^r}
&\sum_{i=1}^{N} \left(\alpha p_i^c +\beta
p_i^r\right), \label{Eq-P3-Obj}\\
\mbox{s.t.}~&~p_i^c\ge 0,~p_i^r\ge 0, \forall i
\in \left\{1,2,\cdots,N\right\},\label{Eq-P3-Con1}\\
&\sum_{i=1}^k p_i^r- \sum_{i=1}^k T_i \le 0, \forall k
\in \left\{1,2,\cdots,N\right\}, \label{Eq-P3-Con2}\\
&\frac{1}{N}\sum_{i=1}^N\chi_i\left(p_i^c,p_i^r\right) \le
\epsilon,\label{Eq-P3-Con3}
\end{align}
\end{pro}
where $\chi_i\left(p_i^c,p_i^r\right)$ is given by \eqref{Eq-Chi}.
For notation convenience, we use $\chi_i$ instead of
$\chi_i\left(p_i^c,p_i^r\right)$ in the rest of the paper. Problem
\ref{Problem-3} is a mixed integer optimization problem, which is
difficult to solve optimally \cite{CombinOp1998}.

For the problem considered here, we assume full CESI, i.e., the
channel power gains (i.e., $[g_1, g_2, \cdots, g_N]^T$) and the
energy harvesting state information (i.e., $[T_1, T_2, $ $\cdots,
T_N]^T$) are known at the Tx as a \emph{priori}. This assumption is
fairly strong and may not be practical. However, the solution
provides a lower bound on the energy cost and sheds insights on the
design of energy allocation strategies with partial CESI where not
all information is available in advance.

\section{Theoretical Results} \label{Sec-TheoResults}
We start by analyzing this problem to obtain structural properties
of the optimal solution, which is useful in developing good
sub-optimal algorithms later.
\subsection{Properties of Problem \ref{Problem-3}}
\begin{prop}\label{Lemma-1}
Denote the set of slots in which the user is in outage as $\mathcal
{S}$. Then, at the optimal solution of Problem \ref{Problem-3}, we
have $|\mathcal {S}^*|=\lfloor N\epsilon \rfloor $, where $|\cdot|$
denotes the cardinality of a set and $\lfloor x \rfloor$ denotes the
largest integer not greater than $x$.
\end{prop}
\begin{proof}
First, any feasible solution of Problem \ref{Problem-3} must satisfy
the constraint \eqref{Eq-P3-Con3}. Thus, we have $|\mathcal
{S}^*|\le\lfloor N\epsilon \rfloor$. Now, suppose $|\mathcal
{S}^*|<\lfloor N\epsilon \rfloor$. Then, we can always drop more
slots such that \eqref{Eq-P3-Con3} holds with equality. Thus, the
energy cost of these slots becomes zero. Obviously, by doing this,
the value of \eqref{Eq-P3-Obj} is reduced. This contradicts with our
presumption that $|\mathcal {S}^*|<\lfloor N\epsilon \rfloor$. Thus,
$|\mathcal {S}^*|$ must be equal to $\lfloor N\epsilon \rfloor $.
\end{proof}

Proposition \ref{Lemma-1} indicates that the optimal allocation
strategy is to drop as many slots as allowed by the outage
constraint. In those dropped slots, the user should shut down its
transmission, and thus consumes no energy.

We next consider the case where the set of outage slots is fixed,
and determine the optimal power allocation policy under this
condition. We first state a Lemma that may be of independent
technical interest.
\begin{lem}\label{Lemma-0}
An optimal policy for the following linear program
\begin{align}
\min_{x_i}~
&\alpha\sum_{i=1}^N c_i-(\alpha-\beta) \sum_{i=1}^N x_i,\\
\mbox{s.t.}~~ &~0\le x_i \le c_i, \forall i \in \{1,2, \ldots, N\},
\\&\sum_{i=1}^k x_i- \sum_{i=1}^k T_i \le 0, \forall
k \in \left\{1,2,\cdots,N\right\}.
\end{align}
is given by $x_i^* = \min\{c_i, \sum_{j=1}^{i}T_j - \sum_{j=1}^{i-1}
x^*_{j}\}$ for $i \in \{1, 2, \ldots, N\}$.
\end{lem}

Proof of this lemma follows from observing that the policy satisfies
the Karush-Kuhn-Tucker (KKT) conditions~\cite{Convexoptimization}
for the linear program. Since the KKT conditions are sufficient for
optimality of linear programs~\cite{Convexoptimization}, this policy
is optimal. The proof is given in the Part A of the Appendix.

\begin{prop}\label{Lemma-2}

For any given set $\mathcal{S}$, the power allocation strategy
$\boldsymbol{p}_k^*=[p_k^r,~p_k^c]^T$ given below is optimal.
\begin{align}\label{Eq-Greedy}
\boldsymbol{p}_k^*=\left\{\begin{array}{ll}
                          \hat{\boldsymbol{p}}_k^*,  &~ \forall k \in
                          \mathcal{S},\\
                          \tilde{\boldsymbol{p}}_k^*,  &~ \forall k \in \mathcal {S}^c,
                        \end{array}
\right.
\end{align}
where $\hat{\boldsymbol{p}}_k^*=[0,~0]^T$,
$\tilde{\boldsymbol{p}}_k^*=\left[p_k^{r*}, p_i^{inv}-
p_k^{r*}\right]^T$ with $p_k^{r*}=\min\left\{p_i^{inv},\sum_{i=1}^k
T_i-\sum_{i=1}^{k-1}p_i^r\right\}$, and $\mathcal {S}^c$ denotes the
complement of $\mathcal {S}$.
\end{prop}


\begin{proof}
It is observed that if $\mathcal {S}$ is given, Problem
\ref{Problem-3} can be converted to
\begin{pro}\label{Problem-Complement}
\begin{align}
\min_{p_i^c,~p_i^r}  &\left(\alpha \sum_{ i \in \mathcal {S}^c}
p_i^c +\beta \sum_{ i
\in \mathcal {S}^c} p_i^r\right),\\
\mbox{s.t.}~&~p_i^c\ge 0,~p_i^r\ge 0,  ~i
\in \mathcal {S}^c,\\
&\sum_{i=1}^k p_i^r- \sum_{i=1}^k T_i \le 0, \forall k
\in \left\{1,2,\cdots,N\right\}, \\
&\ln(1+\frac{g_i\left(p_i^c + p_i^r\right)}{N_0})\ge R, \forall i
\in \mathcal {S}^c, \label{EQ-PComplement-Con3}
\end{align}
\end{pro}
where $\mathcal {S}^c$ denotes the complement of $\mathcal {S}$.

It is not difficult to observe that \eqref{EQ-PComplement-Con3} is
equivalent to $ p_i^c+p_i^r\ge p_i^{inv},\forall i \in \mathcal
{S}^c. $ Obviously, the objective function is minimized when it
holds with equality for all $i \in \mathcal {S}^c$, i.e.,
$p_i^c+p_i^r=p_i^{inv},\forall i \in \mathcal {S}^c.$ Furthermore,
it is also easy to see that our optimization problem is equivalent
to setting $p^{inv}_i = 0$ for all $i \in \mathcal{S}$. Based on
these observations, Problem \ref{Problem-Complement} can be
converted to
\textcolor[rgb]{0.00,0.00,0.00}{\begin{pro}\label{Proof-Problem}
\begin{align}
\min_{p_i^r}~
&\alpha\sum_{i \in \mathcal {S}^c} p_i^{inv}-(\alpha-\beta) \sum_{i \in \mathcal {S}^c} p_i^r,\\
\mbox{s.t.}~~ &~0\le p_i^r \le p_i^{inv}, \forall i \in \mathcal
{S}^c,
\\&\sum_{i=1}^k p_i^r- \sum_{i=1}^k T_i \le 0, \forall
k \in \left\{1,2,\cdots,N\right\}.
\end{align}
\end{pro}}
We note that Problem 3 has the same structure as the linear program
in Lemma~\ref{Lemma-0}. Applying Lemma~\ref{Lemma-0} to with $c_i =
p_{i}^{inv}$ and $x_i=p_i^r$ then concludes our proof of
Proposition~\ref{Lemma-2}.
\end{proof}

Proposition \ref{Lemma-2} indicates that the power allocation is
zero for both conventional and harvested energy in dropped slots.
This is clearly optimal in terms of energy saving. It is also
observed that for the remaining slots, the harvested energy should
be used first. If the harvested energy is not enough to support the
target rate during these slots, conventional energy should be used
as a compensation. This is similar to the greedy use of harvested
energy whenever possible, and thus highlights the fundamental
difference of prioritizing the use of (cheap) harvested energy over
(expensive) conventional energy. Thus, for convenience, we refer to
the power allocation given in \eqref{Eq-Greedy} as
\textbf{\emph{Greedy Power Allocation}}.

\subsection{Optimal Power Allocation Algorithms}
From the results obtained in Proposition \ref{Lemma-1} and
Proposition \ref{Lemma-2}, it is observed that Problem
\ref{Problem-3} in general can be solved in two steps: (1) Find the
set of the time slots that should be dropped, i.e., $\mathcal {S}$.
(2) Apply the greedy power allocation given in \eqref{Eq-Greedy} for
the time slots in $\mathcal {S}^c$.  The difficulty of Problem
\ref{Problem-3} lies primarily in the first step, i.e., find the
optimal $\mathcal {S}^*$ and its complement $(\mathcal {S}^*)^c$. We
start our analysis from two extreme cases ($\lfloor N\epsilon
\rfloor=1$ and $\lfloor N\epsilon \rfloor=N-1$) and then extend the
results to the general case ($\lfloor N\epsilon \rfloor=M$).

\emph{1)} $\lfloor N\epsilon \rfloor=1$: This case is to drop one
slot.

\begin{thm}\label{Theorem-5}
When $\lfloor N\epsilon \rfloor=1$, for any slot
$i\in\{1,2,\cdots,N\}$, if there is a slot before slot $i$ requiring
more channel inversion power than slot $i$, then slot $i$ should be
kept, i.e, $i\in (S^*)^{c}$.
\end{thm}
\begin{proof}
Suppose that there is a slot $k$, where $k<i$, and $p_k^{inv} >
p_i^{inv}$. Now, we consider the following two scenarios:

Scenario 1: Slot $i$ is dropped. For the convenience of exposition,
we denote the total energy available at slots $k$ and $k+1$ are
$E_k$ and $E_{k+1}$, respectively. Then, the harvested energy
consumed during slot $k$ is $p_k^r=E_k-E_{k+1}$. Then, the energy
cost generated by this slot is $\alpha p_k^c+\beta p_k^r$, which is
equivalent to $\alpha p_k^{inv}-(\alpha-\beta) (E_k-E_{k+1})$.

Scenario 2: Slot $k$ is dropped. Thus, no energy is consumed in slot
$k$ and $p_k^{inv}$ in case 1 can be saved for future use.  If
$E_k-E_{k+1}\ge p_i^{inv}$, the energy cost generated by slot $i$ is
$\beta p_i^{inv}$, which is less than $\alpha p_k^c+\beta p_k^r$. If
$E_k-E_{k+1}<p_i^{inv}$, the energy cost generated by slot $i$ is
$\alpha p_i^{inv}-(\alpha-\beta) p_i^r$. In this case, we let
$p_i^r=E_k-E_{k+1}$. It follows that $\alpha
p_i^{inv}-(\alpha-\beta) (E_k-E_{k+1})<\alpha
p_k^{inv}-(\alpha-\beta) (E_k-E_{k+1})$.

It is observed that the energy cost incurred by slot $i$ is always
no more than that incurred by slot $k$. Besides, in other time
slots, for scenario 2, we can always adopt the power allocation used
in scenario 1. Thus, keeping slot $i$ will never result in a larger
total energy cost than keeping slot $k$. Theorem \ref{Theorem-5} is
thus proved.
\end{proof}


Based on the result given in Theorem \ref{Theorem-5}, we are able to
develop the following algorithm to obtain the optimal power
allocation scheme for Problem \ref{Problem-3}.
\begin{algorithm}[h!]
\caption{Optimal power allocation for Problem \ref{Problem-3} when
$\lfloor N\epsilon \rfloor=1$ }\label{alg:1outofN}
\begin{algorithmic}[1]
\STATE Calculate the channel inversion power required for each time
slot, i.e., $p_i^{inv}= \frac{N_0\left(e^{R}-1\right)}{g_i},\forall
i.$

\STATE Initialize $p_{idx}=N$.

\STATE Initialize a set $\mathcal {A}=\left\{1,2,\cdots,N\right\}$,
and  a  candidate set $\mathcal {B}=\emptyset$.

\WHILE {$p_{idx}>1$}

\STATE Find the time slot that requires the largest channel
inversion power in $\mathcal {A}$ and set the value of its index to
$p_{idx}$.

\STATE Put this time slot into the set $\mathcal {B}$.

\STATE \textcolor[rgb]{0.00,0.00,0.00}{Remove the slot $p_{idx}$ and
all the subsequent slots in $\mathcal {A}$.}

\ENDWHILE

\STATE  For any slot $i\in \mathcal {B}$, except the one with the
largest $p_i^{inv}$, slot $i$ can be removed from $\mathcal {B}$ if
the condition $p_i^{inv}<T_i$ is satisfied.

\STATE \textcolor[rgb]{0.00,0.00,0.00}{Perform exhaustive search
over the remaining candidates in $\mathcal {B}$.}
\end{algorithmic}
\end{algorithm}

\begin{figure}[t]
        \centering
        \includegraphics*[width=8.5cm]{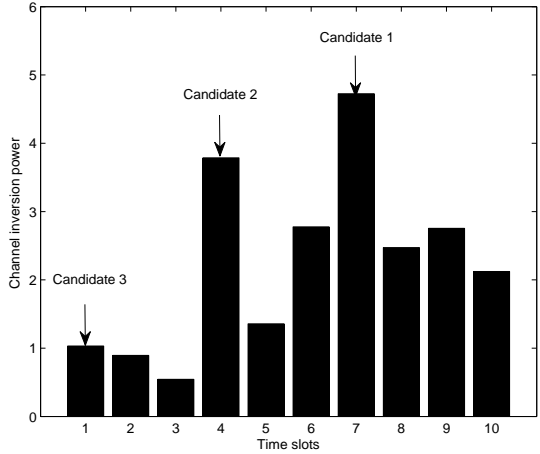}
        \caption{Illustration for Algorithm \ref{alg:1outofN}}
        \label{Fig-illufig1}
\end{figure}




%

We give an example to illustrate Algorithm \ref{alg:1outofN} in Fig.
\ref{Fig-illufig1} with $N=10$. We first put all the  $10$ time
slots into a set $\mathcal {A}$. It is observed that slot $7$
requires most channel inversion power in $\mathcal {A}$. Thus, we
put slot $7$ into the candidate set $\mathcal {B}$ and remove slots
$7$ to $10$ from $\mathcal {A}$. It is observed that slot $4$ now
requires the largest channel inversion power in $\mathcal {A}$.
Therefore, we put slot $4$ into the candidates set $\mathcal {B}$,
and remove slots $4$ to $6$ from $\mathcal {A}$. Then, slot $1$ now
requires most channel inversion power in $\mathcal {A}$. Thus, we
put slot $1$ into the candidates set $\mathcal {B}$ and remove slots
$1$ to $3$ from $\mathcal {A}$. Now, in the candidates set $\mathcal
{B}$, we have
slots $1$, $4$ and $7$. 
It is clear that only the slot with $T_i< p_i^{inv}$ may incur an
energy cost higher than slot $7$ since it has to use some
conventional energy. Thus, for any slot $i$ (except slot $7$) in
$\mathcal {B}$ with $T_i\ge p_i^{inv}$, it can be removed from set
$\mathcal {B}$ since it will not incur an energy cost higher than
slot $7$. In general, after these procedures, the number of the
candidates left in $\mathcal {B}$ is small, and we can easily search
for the optimal solution.

\emph{2)} $\lfloor N\epsilon \rfloor=N-1$: This is equivalent to
keeping one slot and dropping $N-1$ slots. By applying the results
given in Theorem \ref{Theorem-5}, we are able to develop the
following algorithm.

\begin{algorithm}[h]
\caption{Optimal power allocation for Problem \ref{Problem-3} when
$\lfloor N\epsilon \rfloor=N-1$}\label{alg:N1outofN}
\begin{algorithmic}[1]
\STATE Calculate the channel inversion power required for each time
slot, i.e., $p_i^{inv}= \frac{N_0\left(e^{R}-1\right)}{g_i},\forall
i.$

\STATE Initialize $p_{idx}=1$.

\STATE Initialize a set $\mathcal {A}=\left\{1,2,\cdots,N\right\}$,
and  a  candidate set $\mathcal {B}=\emptyset$.

\WHILE {$p_{idx}<N$}

\STATE Find the time slot that requires the smallest channel
inversion power in $\mathcal {A}$ and set the value of its index to
$p_{idx}$.

\STATE Put this time slot into the set $\mathcal {B}$.

\STATE \textcolor[rgb]{0.00,0.00,0.00}{Remove all the time slots
before $p_{idx}$ (including slot $p_{idx}$) in $\mathcal {A}$.}

\ENDWHILE

\STATE For any slot $k\in \mathcal {B}$,  the slots after slot $k$
can be removed from  set $\mathcal {B}$ if $\sum_{i=1}^k
T_i>p_k^{inv}$.

\STATE \textcolor[rgb]{0.00,0.00,0.00}{Perform exhaustive search
over the remaining candidates in $\mathcal {B}$.}
\end{algorithmic}
\end{algorithm}


We give an example to illustrate Algorithm \ref{alg:N1outofN} in
Fig. \ref{Fig-illufig2}. We first put the all $10$ time slots into a
set $\mathcal {A}$. It is observed that slot $3$ requires smallest
channel inversion power in $\mathcal {A}$. Thus, we put slot $3$
into the candidates set $\mathcal {B}$ and remove slots $1$ to $3$
from $\mathcal {A}$. It is observed that slot $5$ now requires
smallest channel inversion power in $\mathcal {A}$. Therefore, we
put slot $5$ into the candidates set $\mathcal {B}$ and remove slots
$4$ to $5$ from $\mathcal {A}$. Then, slot $10$ now requires
smallest channel inversion power in $\mathcal {A}$. Thus, we put
slot $10$ into the candidates set $\mathcal {B}$ and remove slots
$6$ to $10$ from $\mathcal {A}$. Now, in the candidates set
$\mathcal {B}$, we have
slots $3$, $5$ and $10$. 
Then, it is clear that if $\sum_{i=1}^k T_i>p_k^{inv}$, $\forall k
\in \mathcal {B}$, the slots after slot $k$ can be removed from set
$\mathcal {B}$ due to the fact that slot $k$ will not incur an
energy cost higher than those slots after it. In general, after
these procedures, the number of the candidates left in $\mathcal
{B}$ is quite small, and we can easily search for the optimal
solution.

\begin{figure}[t]
        \centering
        \includegraphics*[width=8.5cm]{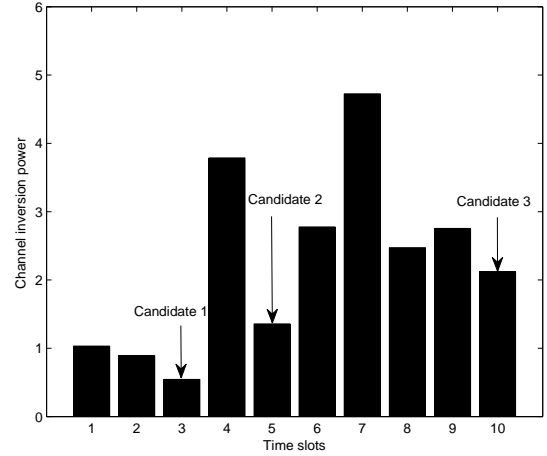}
        \caption{Illustration for Algorithm \ref{alg:N1outofN}}
        \label{Fig-illufig2}
\end{figure}


\emph{3)} $\lfloor N\epsilon \rfloor=M$: This case is to drop $M$
slots.
\begin{thm}\label{Theorem-4}
When $\lfloor N\epsilon \rfloor=M$, for any slot
$i\in\{1,2,\cdots,N\}$, if there are $M$ slots before slot $i$
requiring more channel inversion power than that of slot $i$, then
slot $i$ should not be dropped, i.e, $i\in S^{c}$.
\end{thm}
\begin{proof}
Consider the scenario that there are $M$ slots ahead of slot $i$
requiring more channel inversion power than that of slot $i$.  Now,
we suppose that it is optimal to drop slot $i$. Under these
assumptions, there are two possible cases:
\begin{itemize}
  \item \emph{Case 1: All of those $M$ slots are dropped.} This
  implies a total $M+1$ slots are dropped, which contradicts with
  the fact that $\lfloor N\epsilon \rfloor=M$. Thus, this case
  cannot happen.
  \item \emph{Case 2: $M-1$ or less lots of those $M$ slots are dropped.}
  For this case, there must exist at least one slot $j$ requiring more
  channel inversion power than slot $i$ is not dropped. Then, according to Theorem \ref{Theorem-5}, by
  dropping slot $j$ and keeping slot $i$ instead, we can achieve
  lower energy cost. This contradicts with our presumption
  that it is optimal to drop slot $i$.
\end{itemize}
Combining the above results, it is clear that our presumption does
not hold. By contradiction, slot $i$ should be kept for the scenario
considered. Theorem \ref{Theorem-4} is thus proved.
%
\end{proof}

By applying Theorem \ref{Theorem-4}, we can reduce the number of
channels under consideration and search over the remaining
candidates set.

The proposed optimal power allocation algorithms in this section can
greatly reduce the number of candidates when finding the slots to be
dropped, especially when $\lfloor N\epsilon \rfloor=1$ and $\lfloor
N\epsilon \rfloor=N-1$. However, for the general case $\lfloor
N\epsilon \rfloor=M$, the cardinality of the candidates set is still
large, and thus the complexity of the optimal power allocation
algorithm is quite high. In the following section, we develop two
efficient sub-optimal algorithms.

\subsection{Suboptimal Power Allocation Algorithms}
In this subsection, we propose two suboptimal power allocation
schemes for Problem \ref{Problem-3}, which are given as below.

\subsubsection{Linear Programming based Channel Removal (LPCR)} To develop the first algorithm, we consider the following
problem
\begin{pro}\label{Problem-4}
\begin{align}
\min_{p_i^c,~p_i^r,~\chi_i} &\sum_{i=1}^{N} \left(\alpha p_i^c
+\beta
p_i^r\right),\label{Eq-P4-Obj}\\
\mbox{s.t.}~&~p_i^c\ge 0,~p_i^r\ge 0,~\forall i,\label{Eq-P4-Con1}\\
&\sum_{i=1}^k p_i^r- \sum_{i=1}^k T_i \le 0, \forall k
\in \left\{1,2,\cdots,N\right\}, \label{Eq-P4-Con2}\\
&\sum_{i=1}^N\chi_i \le \lfloor N\epsilon \rfloor,
\label{Eq-P4-Con3}\\
&p_i^c+p_i^r\ge p_i^{inv}
\left(1-\chi_i\right), ~\forall i \label{Eq-P4-Con4},\\
&\chi_i\in\left\{0,1\right\}, ~\forall i.
\end{align}
\end{pro}
Problem \ref{Problem-4} is a mixed-integer programming problem, and
it is easy to verify that Problem \ref{Problem-4} is equivalent to
Problem \ref{Problem-3}. Details are omitted here for brevity.

By taking $\chi_i$ as a continuous variable over $[0,1]$ instead of
a binary variable, the relaxation problem of Problem \ref{Problem-4}
is given by
\begin{pro}\label{Problem-Relaxtion}
\textbf{\emph{Lower bound of Problem \ref{Problem-4}}}
\begin{align}
\min_{p_i^c,~p_i^r,~\chi_i} &\sum_{i=1}^{N} \left(\alpha p_i^c
+\beta
p_i^r\right),\\
\mbox{s.t.}~&~\eqref{Eq-P4-Con1},\eqref{Eq-P4-Con2},\eqref{Eq-P4-Con3},\eqref{Eq-P4-Con4},\\
&~~0\le \chi_i \le 1, ~\forall i.
\end{align}
\end{pro}
Problem \ref{Problem-Relaxtion} is a linear programming problem, and
hence, it can be solved efficiently \cite{cvx}. It is worth pointing
out that Problem \ref{Problem-Relaxtion} provides us a lower bound
to Problem \ref{Problem-4}. Thus, it can be used as a benchmark to
investigate the performance of the proposed suboptimal algorithms.

Based on the results of Problem \ref{Problem-Relaxtion}, the
following suboptimal algorithm for solving Problem \ref{Problem-4}
is developed.

\begin{algorithm}[htb]
\caption{Linear Programming based Channel Removal (LPCR)}
\label{alg:HALP}
\begin{algorithmic}[1]

\STATE Solve Problem \ref{Problem-Relaxtion} by existing linear
programming solvers such as CVX \cite{cvx}.

\STATE Sort the solution
$\left[\chi_1,\chi_2,\cdots,\chi_N\right]^T$ in descending order,
and drop the first $\lfloor N\epsilon \rfloor$ slots.

\STATE Apply the greedy power allocation scheme given in
\eqref{Eq-Greedy} for the remaining time slots.
\end{algorithmic}
\end{algorithm}

\textcolor[rgb]{0.00,0.00,0.00}{The complexity analysis of Algorithm
\ref{alg:HALP} is given as follows. The worst-case complexity of
solving the linear program in step $1$ is $O(N^3)$ (see
\cite{CombinOp1998}). The complexity of sorting the obtained
solution in descending order in step $2$ is $O(N\log N)$ (see
\cite{Intro2Algo}). The complexity of step $3$ is $O(N)$. Thus, the
complexity of LPCR is $O(N^3)$.}



\subsubsection{Worst-Channel Removal (WCR)}
\textcolor[rgb]{0.00,0.00,0.00}{Although the LPCR algorithm has
polynomial time complexity, it still requires solving a linear
program (Problem \ref{Problem-Relaxtion}) with complexity $O(N^3)$.
In this subsection, we propose a simpler suboptimal algorithm,
referred to as \emph{Worst-Channel Removal} (WCR), which has a worst
case complexity of $O(N \log N)$.}

\begin{algorithm}[htb]
\caption{Worst-Channel Removal (WCR)} \label{alg:WCRA}
\begin{algorithmic}[1]

\STATE Sort the time slots according to their channel power gains in
the descending order.


\STATE Drop the first $\lfloor N\epsilon \rfloor$ slots.

\STATE Apply the greedy power allocation scheme given in
\eqref{Eq-Greedy} for the remaining time slots.

\end{algorithmic}
\end{algorithm}


The idea of WCR is to remove the worst $\lfloor N\epsilon \rfloor$
channels. It is clear that WCR is in general not optimal. However,
when certain conditions are satisfied, WCR is optimal. In the
following, we investigate three conditions when WCR is optimal,
hence strengthening the motivation of using WCR as a heuristic
scheme.

\begin{thm} \label{Theorem-2}
WCR is the optimal solution of Problem \ref{Problem-3}  if the
condition $\sum_{i=1}^{N} p_i^r- \sum_{i=1}^{N} T_i=0$ is satisfied,
i.e., the harvested energy is fully consumed at the end of the
transmission.
\end{thm}

\begin{proof}
%
%
%
Let $\mathcal {S}$ be the set given in WCR, and we assume that WCR
satisfies the condition $\sum_{i=1}^{N} p_i^r- \sum_{i=1}^{N}
T_i=0$. Then, according to the proof of
\textcolor[rgb]{0.00,0.00,0.00}{Lemma 1}, for a given set, Problem
\ref{Problem-3} can be converted to Problem \ref{Proof-Problem}.
Thus, under the above assumptions, the value of the objective
function under WCR is
\begin{align}\label{EQ-T1Proof-1}
\alpha \sum_{i \in \mathcal {S}^c} p_i^{inv}
-\left(\alpha-\beta\right) \sum_{i=1}^N T_i.
\end{align}

Let $\hat{\mathcal {S}}$ be any feasible solution set (other than
$\mathcal {S}$) of Problem \ref{Problem-3}, the value of the
objective function under $\hat{\mathcal {S}}$ is then given by
\begin{align}\label{EQ-T1Proof-2}
\alpha \sum_{i \in \hat{\mathcal {S}}^c} p_i^{inv}
-\left(\alpha-\beta\right) \sum_{i \in \hat{\mathcal {S}}^c}  p_i^r,
\end{align}
where $\hat{\mathcal {S}}^c$ is the complement of $\hat{\mathcal
{S}}$. Since $\hat{\mathcal {S}}$ is a feasible solution of Problem
\ref{Problem-3}, it can be observed that $\sum_{i \in \hat{\mathcal
{S}}^c}  p_i^r \le \sum_{i=1}^N T_i$. Since ${\mathcal {S}}$
contains the $\lfloor N\epsilon \rfloor$ time slots with weakest
channel power gains, it is easy to verify that $\sum_{i \in \mathcal
{S}^c} p_i^{inv} \le \sum_{i \in \hat{\mathcal {S}}^c} p_i^{inv}$.
Based on these observations, \textcolor[rgb]{0.00,0.00,0.00}{it is
clear that \eqref{EQ-T1Proof-1} is always lower than
\eqref{EQ-T1Proof-2}}. Thus, Theorem \ref{Theorem-2} is proved.
\end{proof}

\textcolor[rgb]{0.00,0.00,0.00}{This theorem can be explained in the
following way. For any resource allocation schemes that consume all
the harvested energy, the cost of the renewable energy for these
schemes is the same. Thus, the cost difference among these schemes
comes from the cost of the conventional energy. Thus, the scheme
consuming less conventional energy has a lower total energy cost.
Therefore, it is clear that WCR is optimal when all the harvested
energy is consumed.}

\begin{thm} \label{Theorem-3}
WCR is the optimal solution of Problem \ref{Problem-3} if no
conventional energy is consumed during  the whole transmission
process.
\end{thm}
\begin{proof}
Let $\mathcal {S}$ be the set given in WCR, and it satisfies the
condition that no conventional energy is consumed during the whole
transmission process. Since no conventional energy is consumed, the
value of the objective function of Problem \ref{Problem-3} under WCR
is given by
\begin{align}\label{Eq-T3Proof-1}
\sum_{i \in \mathcal {S}^c} \beta
p_i^r,~\mbox{where}~p_i^r=p_i^{inv}, \forall i \in \mathcal
{S}^c.\end{align}

Let $\hat{\mathcal {S}}$ be any feasible solution set (other than
$\mathcal {S}$) of Problem \ref{Problem-3}, the value of the
objective function under $\hat{\mathcal {S}}$ is then given by
\begin{align}\label{Eq-T3Proof-2} \sum_{i \in \hat{\mathcal {S}}^c} \left(\alpha p_i^c +\beta
p_i^r\right),~\mbox{where}~p_i^c+p_i^r= p_i^{inv},\forall i \in
\hat{\mathcal {S}}^c.
\end{align}

From the fact that $\alpha>\beta$, it follows that $\sum_{i \in
\hat{\mathcal {S}}^c} \left(\alpha p_i^c +\beta p_i^r\right)>
\sum_{i \in \hat{\mathcal {S}}^c} \beta \left(p_i^c + p_i^r\right)$.
It is easy to verify that $\sum_{i \in \hat{\mathcal {S}}^c} \beta
\left(p_i^c + p_i^r\right)=\sum_{i \in \hat{\mathcal {S}}^c} \beta
p_i^{inv}>\sum_{i \in \mathcal {S}^c} \beta p_i^{inv}$ due to the
fact that ${\mathcal {S}}$ contains the $\lfloor N\epsilon \rfloor$
time slots with weakest channel power gains. Thus, it is clear that
\eqref{Eq-T3Proof-1} is always less than \eqref{Eq-T3Proof-2}.
Theorem \ref{Theorem-3} is thus proved.
\end{proof}

\textcolor[rgb]{0.00,0.00,0.00}{This theorem can be explained in the
following way. For any resource allocation schemes that consume no
conventional energy, the total energy cost is only determined by the
cost of the renewable energy. Thus, the scheme consuming less
renewable energy has a lower total energy cost. Therefore, it is
clear that WCR is optimal when no conventional energy is consumed.}

\begin{thm} \label{Theorem-6}
For any type of non-decreasing (over time) channel (e.g., AWGN
channel), WCR gives the optimal solution of Problem \ref{Problem-3}.
\end{thm}
\begin{proof}
For any type of non-decreasing (over time) channel (e.g., AWGN
channel), WCR is equivalent to dropping the first $\lfloor N\epsilon
\rfloor$ slots.  Let $\mathcal {S}$ be the set that we drop the
first $\lfloor N\epsilon \rfloor$ slots, i.e., $\mathcal
{S}=\left\{g_1, g_2,\cdots,g_{\lfloor N\epsilon \rfloor}\right\}$.
To guarantee the QoS of the user during the remaining time slots,
the transmit power required is $p_i^c+p_i^r= p_i^{inv},\forall i \in
{\mathcal {S}}^c$, where $\mathcal {S}^c=\left\{g_{\lfloor N\epsilon
\rfloor+1}, g_{\lfloor N\epsilon \rfloor+2},\cdots,g_N\right\}$.

Now, we consider the set $\hat{\mathcal {S}}=\left\{
g_2,\cdots,g_{\lfloor N\epsilon \rfloor}, g_{\lfloor N\epsilon
\rfloor+1}\right\}$. To guarantee the QoS of the user during the
remaining time slots, the transmit power required is
$\hat{p_i}^c+\hat{p_i}^r= p_i^{inv},\forall i \in \hat{\mathcal
{S}}^c$, where $\hat{\mathcal {S}}^c=\left\{g_{1}, g_{\lfloor
N\epsilon \rfloor+2},\cdots,g_N\right\}$. Since the channel is
non-decreasing, it is clear that $g_1 \le g_{\lfloor N\epsilon
\rfloor+1}$. Thus, it follows $p_{\lfloor N\epsilon
\rfloor+1}^c+p_{\lfloor N\epsilon \rfloor+1}^r
\le\hat{p_1}^c+\hat{p_1}^r$. For other time slots, we have
$p_{i}^c+p_{i}^r =\hat{p_i}^c+\hat{p_i}^r$. Now, we look at the
energy harvest constraints under $\mathcal {S}$ and $\hat{\mathcal
{S}}$, respectively. Under $\mathcal {S}$, we have $p_{\lfloor
N\epsilon \rfloor+1}^r \le \sum_{i=1}^{\lfloor N\epsilon \rfloor+1}
T_i$. While under $\hat{\mathcal {S}}$, we have $\hat{p_1}^r \le
T_1$. All the remaining energy harvest constraints are exactly the
same for the two cases. Thus, we are always able to set $p_{\lfloor
N\epsilon \rfloor+1}^r =\hat{p_1}^r$ and $p_{i}^r =\hat{p_i}^r$ for
all the remaining time slots. Then, it is observed that the
resultant total energy cost under $\mathcal {S}$ is always less or
equal to that under $\hat{\mathcal {S}}$. Using the same approach,
we can prove that the total energy cost under $\mathcal {S}$ is
lower than that under any other feasible solution set of Problem
\ref{Problem-3}.
%
%
\end{proof}

\textcolor[rgb]{0.00,0.00,0.00}{This theorem can be explained in the
following way. For non-decreasing (over time) channels, the channel
inversion power for latter slots is equal to or lower than that for
the former slots. Besides, the renewable energy available for the
latter slots is in general more than that for the former slots.
Thus, dropping the former slots always results in a lower total
energy cost. Therefore, it is clear that WCR is optimal for any type
of non-decreasing channel. }

\subsection{The Multi-Cycle Scenario}
\label{Sec-extension} In the previous subsections, we consider the
single-cycle scenario, i.e., the outage constraint is imposed over
$N$ continuous slots from one cycle. In this subsection, we consider
the multi-cycle scenario, in which the outage constraint is imposed
on each cycle. We assume that there are $M$ cycles, and each cycle
has $N$ time slots. In each cycle, the maximum number of slots that
can be dropped is $K$. Then, the energy cost minimization problem
with energy harvesting constraints can be formulated as follows.


\begin{pro}\label{Problem-SM} 
\begin{align}
\min_{p_i^c,~p_i^r} &\sum_{i=1}^{MN} \left(\alpha p_i^c +\beta
p_i^r\right),\\
\mbox{s.t.}~&~p_i^c\ge 0,~p_i^r\ge 0, ~\forall i,\\
&\sum_{i=1}^k p_i^r- \sum_{i=1}^k T_i \le 0, ~\forall k
\in \left\{1,2,\cdots,MN\right\}, \\
&\sum_{i=(j-1)N+1}^{(j-1)N+N}\chi_i \le K, ~\forall j\in
\left\{1,2,\cdots,M\right\},\label{eq-45}
\end{align}
\end{pro}
where
\begin{align}
\chi_i=\left\{\begin{array}{ll}
                1, & ~\mbox{if}~\ln(1+\frac{g_i\left(p_i^c +
p_i^r\right)}{N_0})<R_i, \\
                0, & ~\mbox{otherwise}.
              \end{array}
\right..
\end{align}
Since Problem \ref{Problem-3} is a special case of Problem
\ref{Problem-SM}, we expect  the optimal solution of Problem
\ref{Problem-SM} to be hard to obtain. Thus, in this subsection, we
develop two suboptimal algorithms to solve Problem \ref{Problem-SM}
based on the LPCR and WCR developed for the one-cycle case. The
extension from the one-cycle case to the multi-cycle case depends on
an important property of Problem \ref{Problem-SM}, which is
presented in the following proposition.
\begin{prop}\label{Pro-1}
At the optimal solution of Problem \ref{Problem-SM}, the constraints
given by \eqref{eq-45} must hold with equality, i.e.,
$\sum_{i=(j-1)N+1}^{(j-1)N+N}\chi_i=K, ~\forall j\in
\left\{1,2,\cdots,M\right\}$.
\end{prop}

Proposition \ref{Pro-1} can be proved by the same approach as
Proposition \ref{Lemma-1}. Thus, details are omitted here for
brevity.

In the following, we present the multi-cycle LPCR and the
multi-cycle WCR, respectively.

%
%

\subsubsection{Multi-cycle LPCR} Denote the leftover harvested energy
of cycle $i$ as $L_i$, and denote the initial storage energy of
cycle $i$ as $S_i$, we can extend the LPCR to the multi-cycle
scenario, which is given as follows.
\begin{algorithm}[htb]
\caption{Multi-cycle LPCR} \label{alg:MCLPCR}
\begin{algorithmic}[1]
\STATE Initialization: $L_{0}=0$.

\FOR {$i=1:M$}

\STATE \textcolor[rgb]{0.00,0.00,0.00}{Compute the initial energy of
cycle $i$ by $S_i=L_{i-1}$.}

\STATE Solve Problem \ref{Problem-Relaxtion} with initial storage
$S_i$ for cycle $i$ by existing linear programming solvers such as
CVX.

\STATE Sort the obtained $\chi_j,\forall j$ in cycle $i$ in
descending order, and drop the first $K$ slots.

\STATE Apply the greedy power allocation scheme during the rest of
time slots in cycle $i$.

\STATE \textcolor[rgb]{0.00,0.00,0.00}{Compute the leftover
harvested energy $L_i$ at the end of cycle $i$.}

\ENDFOR

\end{algorithmic}
\end{algorithm}

The multi-cycle LPCR algorithm solves the multi-cycle problem cycle
by cycle. We note that it is the leftover harvested energy that
couples the cost minimization problem of different cycles. For
example, if no harvest energy is left for future cycles, the
optimization problem in each cycle can be solved independently. In
the following, we investigate how the leftover harvested energy
affects the energy cost.

\begin{prop}\label{Pro-3}
Define initial storage state to be $S$ ($\ge 0$). Let $\pi(S)$ be
the optimal policy at storage state $S$ and $\pi(S+\Delta)$ be the
optimal policy at storage state $S + \Delta$. Let $V(S)$  and
$V(S+\Delta)$ denote the total energy cost under $S$ and $S+\Delta$,
respectively. Then, the cost difference is bounded by $(\alpha -
\beta)\Delta$, i.e., $V(S)-V(S+ \Delta) \le (\alpha - \beta)\Delta$.
\end{prop}


\begin{proof}
See Part B of the Appendix for details.
\end{proof}

\textcolor[rgb]{0.00,0.00,0.00}{From Proposition \ref{Pro-3}, it is
observed that with additional initial storage of $\Delta$, the
maximum cost that the user can reduce is $(\alpha -\beta)\Delta$.
Thus, if the cost increase in the previous cycles to produce
additional harvested energy $\Delta$ is larger than $(\alpha
-\beta)\Delta$, then it is clear that the resource allocation
strategies in previous cycles will not affect the resource
allocation strategies in the current and the following cycles. This
is because it is the leftover harvested energy that couples the cost
minimization problem of different cycles. Thus, if the condition
given in Proposition \ref{Pro-3} is satisfied for all the cycles,
the optimization problem in each cycle can be solved independently.
Otherwise, Algorithm \ref{alg:MCLPCR} can be used to solve the
problem. }

\subsubsection{Multi-cycle WCR}
%

It is observed that the multi-cycle LPCR algorithm requires solving
a series of linear programming problems, which may incur high
complexity for the worst-case scenario. Thus, in this part, we
develop the multi-cycle WCR, which is implemented by a simpler
suboptimal algorithm.

\begin{algorithm}[htb]
\caption{Multi-cycle WCR} \label{alg:McycleWCRA}
\begin{algorithmic}[1]

\STATE Sort the time slots according to the descending order of
channel power gains in each cycle.

\STATE Drop the worst $K$ time slots in each cycle.

\STATE Apply the greedy power allocation scheme during the rest of
time slots.
\end{algorithmic}
\end{algorithm}

The key idea of the multi-cycle WCR is to remove the worst $K$
channels in each cycle. It is clear that multi-cycle WCR is in
general not optimal. However, when certain conditions are satisfied,
the multi-cycle WCR gives the optimal solution. In the following, we
give out a sufficient condition for the multi-cycle WCR to be
optimal.

%


%
\begin{prop}\label{Pro-4}
If WCR is the optimal solution for each individual cycle $i,
~\forall i$, then the multi-cycle WCR is the optimal solution for
Problem \ref{Problem-SM}.
\end{prop}

\begin{proof}
See Part C of the Appendix for details.
\end{proof}

\section{Optimal Resource Allocation \textcolor[rgb]{0.00,0.00,0.00}{With Partial CESI}}
\label{Sec-NoCSI} \textcolor[rgb]{0.00,0.00,0.00}{In previous
sections, we assume full CESI, i.e., the channel power gains (i.e.,
$[g_1, g_2, \cdots, g_N]^T$) and the energy harvesting state
information (i.e., $[T_1, T_2, $ $\cdots, T_N]^T$) is a
\emph{priori} known at the Tx. In this section, we consider the
scenario that only the channel fading statistics (i.e., partial
CESI) are available at the Tx. }Under this assumption, we model the
QoS criterion by the following the equation.
\begin{align}
\mbox{Prob}\left\{\ln(1+\frac{g_i\left(p_i^c +
p_i^r\right)}{N_0})<R\right\} \le \epsilon, \forall i \in
\left\{1,2,\cdots,N\right\}.
\end{align}
This constraint requests the outage probability of the user's
transmission in each time slot to be less than or equal to
$\epsilon$. Thus, the outage probability of the whole transmission
process is less than or equal to $\epsilon$. With this constraint,
the energy cost minimization problem is formulated as
\begin{pro}\label{Problem-1}
\begin{align}\label{P1-Obj}
\min_{p_i^c,~p_i^r} &\sum_{i=1}^{N} \left(\alpha p_i^c +\beta
p_i^r\right),\\
\mbox{s.t.}~&~p_i^c\ge 0,~p_i^r\ge 0,\forall i
\in \left\{1,2,\cdots,N\right\}\\
&\sum_{i=1}^k p_i^r- \sum_{i=1}^k T_i \le 0, \forall k
\in \left\{1,2,\cdots,N\right\}, \\
&~\mbox{Prob}\left\{\ln(1+\frac{g_i\left(p_i^c +
p_i^r\right)}{N_0})<R\right\} \le \epsilon, \nonumber\\&~~\forall i
\in \left\{1,2,\cdots,N\right\}. \label{P1-Con3}
\end{align}
\end{pro}
\textcolor[rgb]{0.00,0.00,0.00}{It is observed that we assume the
energy harvesting state information (i.e., $[T_1, T_2, $ $\cdots,
T_N]^T$) is known at the Tx in this problem formulation. However, it
is worthy pointing out that the future energy harvesting state
information is in fact not required to obtain the optimal solution,
which is given in the following theorem. That is, the optimal power
allocation $p_k^*$ in slot $k$ does not depend on $T_{k+1}, \cdots,
T_N$.}
\begin{thm}\label{Theorem-1}
The optimal solution $\boldsymbol{p}_k^*=[p_k^r,~p_k^c]^T$ for
Problem \ref{Problem-1} is given by
\begin{align}\label{eq-TxPower}
\boldsymbol{p}_k^*=\left\{\begin{array}{ll}
                          \tilde{\boldsymbol{p}}_k^*,  &~ \mbox{if}~
\frac{N_0\left(e^R-1\right)}{F^{-1}\left(\epsilon\right)}\le
\sum_{i=1}^k T_i-\sum_{i=1}^{k-1} p_i^r \\
                          \hat{\boldsymbol{p}}_k^*,  &~ \mbox{if}~
\frac{N_0\left(e^R-1\right)}{F^{-1}\left(\epsilon\right)}>
\sum_{i=1}^k T_i-\sum_{i=1}^{k-1} p_i^r
                        \end{array}
\right., \forall k, 
\end{align}
where
$\tilde{\boldsymbol{p}}_k^*=[\frac{N_0\left(e^R-1\right)}{F^{-1}\left(
\epsilon\right)}~~ 0]^T$, and
$\hat{\boldsymbol{p}}_k^*=[\sum_{i=1}^k T_i-\sum_{i=1}^{k-1} p_i^r~~
\frac{N_0\left(e^R-1\right)}{F^{-1}\left(
\epsilon\right)}-\left(\sum_{i=1}^k T_i-\sum_{i=1}^{k-1}
p_i^r\right)]^T$.
\end{thm}

\begin{proof}
It is east to observe that the constraints given in \eqref{P1-Con3}
are equivalent to
\begin{align}
\mbox{Prob}\left\{g_i<\frac{N_0\left(e^R-1\right)}{p_i^c +
p_i^r}\right\} \le \epsilon, \forall i.
\end{align}
Since the distribution of $g_i$'s is i.i.d and with the CDF
$F(\cdot)$, then it is easy to observe that
$\mbox{Prob}\left\{g_i<\frac{N_0\left(e^R-1\right)}{p_i^c +
p_i^r}\right\}=F\left(\frac{N_0\left(e^R-1\right)}{p_i^c +
p_i^r}\right)$. Consequently, we have
\begin{align}\label{Eq-1}
F\left(\frac{N_0\left(e^R-1\right)}{p_i^c + p_i^r}\right) \le
\epsilon, \forall i.
\end{align}
Define that $g(p_i)\triangleq\frac{N_0\left(e^R-1\right)}{p_i}$,
where $p_i=p_i^c + p_i^r$. Since the CDF function $F(g(p_i))$ is an
increasing function with respect to $g(p_i)$, and $g(p_i)$ is a
decreasing function with respect to $p_i$, it can be inferred that
$F(g(p_i))$ is a decreasing function with respect to $p_i$. Thus,
\eqref{Eq-1} can be converted to
\begin{align} \label{Eq-2}
p_i^c + p_i^r \ge\frac{N_0\left(e^R-1\right)}{F^{-1}\left(
\epsilon\right)}, \forall i,
\end{align}
where $F^{-1}(\cdot)$ is the inverse function of  $F(\cdot)$. Thus,
to minimize the power consumption,  \eqref{Eq-2} should hold with
equality for each $i$, i.e., $p_i^c + p_i^r
=\frac{N_0\left(e^{R}-1\right)}{F^{-1}\left( \epsilon\right)},
\forall i. $ Then, it follows that $p_i^r
\le\frac{N_0\left(e^{R}-1\right)}{F^{-1}\left( \epsilon\right)},
\forall i, $ and \textcolor[rgb]{0.00,0.00,0.00}{$ \sum_{i=1}^N
p_i^c + \sum_{i=1}^N p_i^r
=\sum_{i=1}^N\frac{N_0\left(e^{R}-1\right)}{F^{-1}\left(
\epsilon\right)}. $} Based on these equations, Problem
\ref{Problem-1} can be converted to
\begin{align}
\min_{p_i^r}~
&\alpha\sum_{i=1}^N\frac{N_0\left(e^{R}-1\right)}{F^{-1}\left(
\epsilon\right)}-(\alpha-\beta) \sum_{i=1}^{N} p_i^r,\label{P2-Obj}\\
\mbox{s.t.}~~ &~0\le p_i^r \le
\frac{N_0\left(e^{R}-1\right)}{F^{-1}\left( \epsilon\right)},
\forall i,
\\&\sum_{i=1}^k p_i^r- \sum_{i=1}^k T_i \le 0, \forall
k \in \left\{1,2,\cdots,N\right\}.
\end{align}
This problem has the same structure as the linear optimization
problem in Lemma~\ref{Lemma-0}. Applying Lemma~\ref{Lemma-0} with
$c_i=\frac{N_0\left(e^{R}-1\right)}{F^{-1}\left( \epsilon\right)}$
and $x_i=p_i^r$ then concludes the proof.
\end{proof}

\section{Numerical Results} \label{Sec-NumericalResults}
In this section, we present several numerical examples to evaluate
the performance of the proposed optimal and suboptimal algorithms.

\subsection{Simulation setup}  In the simulation, the target transmission rate $R$ of the user is set to one. The receiver noise power $N_0$
is also assumed to be one. Unless specifically declared, we assume
i.i.d. Rayleigh fading for all channels. Thus, the channel power
gains are exponentially distributed, and we assume that the mean of
the channel power gain is one. The conventional energy is priced at
$\alpha=1$ per unit, and the harvested energy is priced at
$\beta=0.2$ per unit. The incoming energy $T_i$ is modeled as a
random variable with uniform distribution over the range $[0 ~1]$,
i.e., $T\sim\mathcal {U}(0,1)$. In practice, the characteristics of
the incoming energy depends on the type of renewable energy source.
For example, it is shown in \cite{hoICCS10} that the energy can be
modeled as a Markovian chain
with memory. 
For a given type of
 energy harvester, the characteristics of the incoming energy can be
 obtained through long-term measurements. It is worth pointing out
 that the assumption of particular distributions of the channel power gains
 and the incoming energy does not change the structure of the
 problems studied and the algorithms proposed in this paper.

\subsection{Extreme Cases: $\lfloor N\epsilon \rfloor=1$ and $\lfloor N\epsilon \rfloor=N-1$}
\begin{figure}[t]
        \centering
        \includegraphics*[width=9cm]{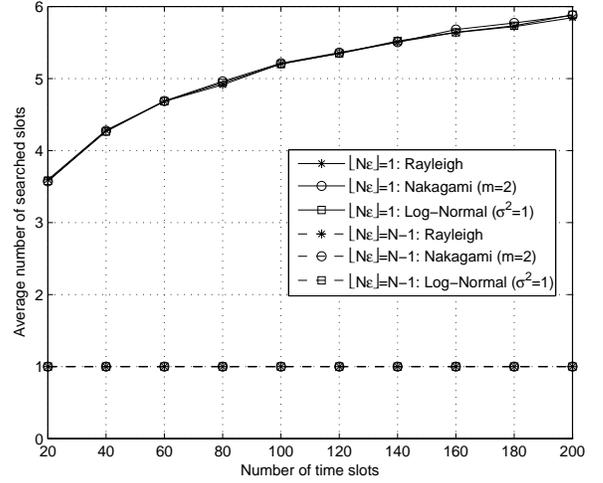}
        \caption{Average number of searched slots for the proposed algorithm under different fading scenarios}
        \label{Fig-kxfig1number}
\end{figure}
\textcolor[rgb]{0.00,0.00,0.00}{In Fig. \ref{Fig-kxfig1number}, we
investigate the average number of searched slots of Algorithm 1 and
Algorithm 2, respectively, under different fading scenarios.  In the
simulation, $m = 2$ is chosen for the unit-mean Nakagami fading
channels used. For log-normal fading channels, $\sigma^2=1$ is used,
with which the dB-spread will be within its typical ranges
\cite{kangTWC}. The average is taken over $10000$ channel
realizations. It is observed from Fig. \ref{Fig-kxfig1number} that
the average number of slots that the proposed algorithm has searched
is almost the same for different fading scenarios. Besides, it is
observed from Fig. \ref{Fig-kxfig1number} that the average number of
searched slots of Algorithm 1 increases when the number of total
slots increases. However, it increases at a very slow rate. It can
be seen that the average number of searched slots is $6$ when the
total number of slots is $200$.  The average number of searched
slots that Algorithm 2 has to search turns out to be one regardless
of the number of total slots. This is due to the fact that since we
only need to keep one slot, in most scenarios, the accumulated
harvested energy is enough to support the channel inversion power of
the slot with the highest channel power gains. For both cases, when
the exhaustive search is adopted, the number of searched slots is
$200$. This indicates that the proposed algorithms are highly
efficient as compared to the exhaustive search whose complexity is
linear.}

\subsection{General Case: $\lfloor N\epsilon \rfloor=M$}
\label{sec-M}

\begin{figure}[t]
        \centering
        \includegraphics*[width=8.5cm]{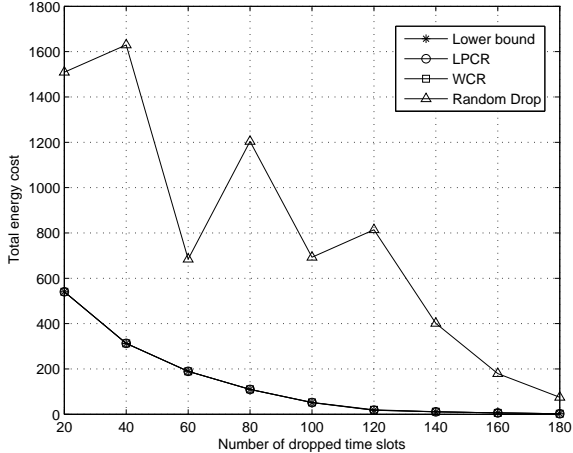}
        \caption{Comparison of the suboptimal algorithms with the lower bound}
        \label{Fig-myfig3a}
\end{figure}

In this subsection, we consider the case that there are $200$ time
slots. In Fig. \ref{Fig-myfig3a}, we plot the total energy cost vs.
the number of dropped slots. The result is obtained by averaging
over $1000$ channel realizations.
\textcolor[rgb]{0.00,0.00,0.00}{The \emph{random drop} algorithm, in
which we drop the slots randomly and apply the greedy power
allocation in the remaining time slots, is given as a baseline
policy. From Fig. \ref{Fig-myfig3a}, it is observed that the
performance of the random drop is the worst. This indicates that
optimization contributes to significant energy saving for our
problem. It is also observed that both the proposed suboptimal
schemes, namely the LPCR and the WCR, can achieve almost the same
performance as the lower bound. Furthermore, it is observed from
Fig. \ref{Fig-myfig3a} that the total energy cost decreases as the
number of dropped slots increases for LPCR, WCR and the lower-bound,
which is as expected. However, for the random drop, this does not
hold. For example, the total energy cost of dropping $60$ slots can
be lower than that of dropping $80$ slots. This can be explained as
follows. In the random drop, since the slots are dropped randomly,
it is possible that most of the $80$ dropped slots are with good
channels and most of the $60$ dropped slots are with bad channels.
As a result, the energy cost of dropping $60$ slots may be lower
than that of dropping $80$ slots.}

\begin{figure}[t]
        \centering
        \includegraphics*[width=8.5cm]{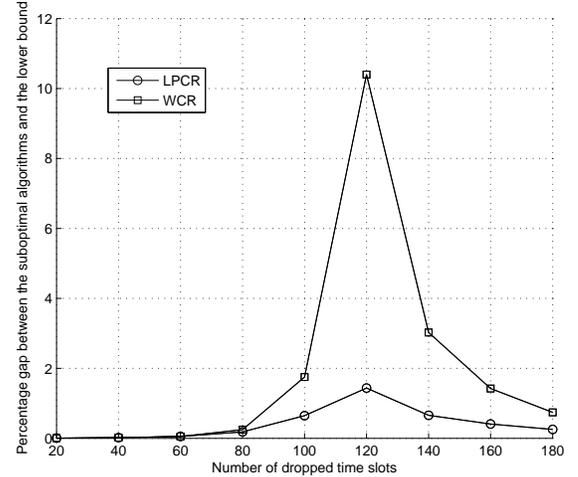}
        \caption{Gaps between the suboptimal algorithms and the lower bound}
        \label{Fig-myfig3}
\end{figure}

In Fig. \ref{Fig-myfig3}, we plot the gap between the suboptimal
schemes and the lower bound vs. the number of dropped slots. The
result is obtained by averaging over 1000 channel realizations. It
is observed that LPCR performs better than WCR in general. However,
WCR performs close to the lower-bound when the number of slots to be
dropped is either small (less than 80) or large (more than 160). For
intermediate range (about 120 slots), WCR has about a $10\%$ gap
from the lower bound. The intuition is as follows. When there are
$200$ slots, it is likely that there will be a small number of slots
that are in deep fading. At the optimal solution, these slots are
likely to be dropped as the cost required to serve these slots is
very high, even if all of them are using the cheap renewable energy.
Similarly, it is also likely that there will be a number of slots in
which the channel power gain is high, and these slots are likely to
be kept as only a small amount of energy is needed to serve these
slots. Since WCR is a greedy heuristic in which channels with low
power gains are dropped, it is likely to agree with the optimal
solution when the number of slots to be dropped is small or large.
However, at the intermediate range, in addition to dropping the
channels in deep fading and keeping channels with high power gains,
we also have to make a decision on channels with moderate power
gains. WCR only drops the channels with lower power gains, but does
not take into account the renewable energy supply pattern. For
channels with moderate power gains, the variation in the renewable
energy supply may result in channels with lower power gains being
kept in the optimal solution. This results in the sub-optimality of
WCR. In contrast, in the LPCR algorithm, we try to take into account
the variation in renewable energy through linear programming
relaxation, resulting in a better performance in the intermediate
range compared to WCR.

\subsection{The multi-cycle case}
\begin{figure}[t]
        \centering
        \includegraphics*[width=8.5cm]{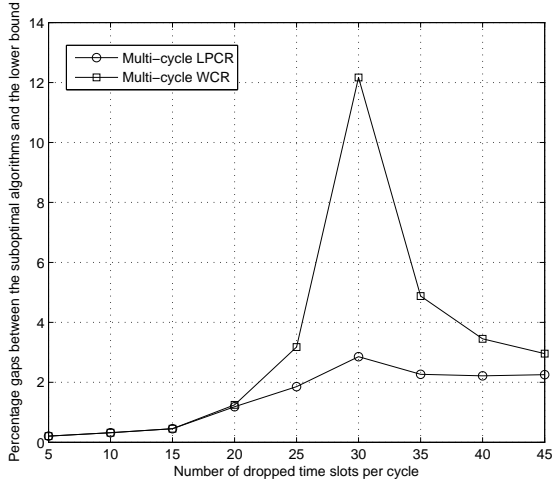}
        \caption{Gaps between the suboptimal algorithms and the lower bound: Multi-cycle case}
        \label{Fig-myfig4}
\end{figure}
\textcolor[rgb]{0.00,0.00,0.00}{In this case, we consider the case
where there are in total $200$ time slots. These time slots are
divided into $4$ cycles, and thus each cycle contains $50$ time
slots. In Fig. \ref{Fig-myfig4}, we plot the gap between the
suboptimal schemes and the lower bound vs. the number of dropped
slots in each cycle. The result is obtained by averaging over 1000
channel realizations. It is observed that the shape of the curves in
this figure is similar to that of the curves in Fig.
\ref{Fig-myfig3}. Multi-cycle LPCR in general performs better than
multi-cycle WCR. Multi-cycle WCR performs close to the lower-bound
when the number of slots to be dropped is either small or large.
This can be explained in the same way as the single-cycle case given
in Section \ref{sec-M}.}

\subsection{The partial CESI case}
\begin{figure}[t]
        \centering
        \includegraphics*[width=8.5cm]{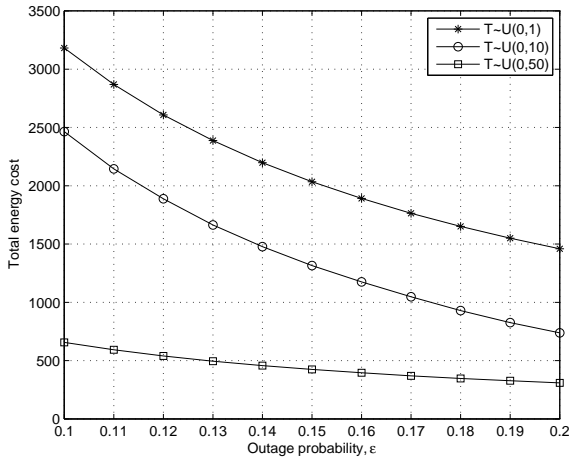}
        \caption{Performance for the proposed resource allocation scheme with partial CESI}
        \label{Fig-myfig5}
\end{figure}
\textcolor[rgb]{0.00,0.00,0.00}{In this subsection, we investigate
the performance for the proposed resource allocation scheme with
partial CESI. We consider the case that the channel fading
statistics (i.e., exponentially distributed with mean $1$) are
available at the Tx, while the energy harvesting state information
is not known at the Tx. In Fig. \ref{Fig-myfig5}, we plot the total
energy cost of the proposed scheme vs. the given outage probability
(i.e., $\epsilon$) under different energy harvesting profiles. It is
observed that the total energy cost of all curves decrease with the
increase of $\epsilon$. This is as expected since the transmit power
is in general inversely proportional to $\epsilon$, which can be
observed from \eqref{eq-TxPower}. Another observation is that for
the same $\epsilon$, the total energy cost under $T\sim\mathcal
{U}(0,50)$ is lower than that under $T\sim\mathcal {U}(0,10)$. This
indicates that the harvested energy plays a significant role in
determining the energy cost.}

\section{Conclusions}\label{Sec-Conclusions}
In this paper, we have considered the problem of communicating over
a block fading channel in which the transmitter has access to an
energy harvester and a conventional energy source, and sought to
minimize the total energy cost of the transmitter, subject to an
outage constraint. This problem is shown to be a mixed integer
programming problem. Optimal algorithms with worst case linear time
complexity have been obtained in two extreme cases: when the number
of slots in outage is $1$ or $N-1$. For the general case  of
allowing $1<k<N-1$ slots in outage, using a linear programming
relaxation, we have obtained an efficiently computable lower bound
as well as a suboptimal algorithm (upper bound), Linear Programming
based Channel Removal (LPCR), for this problem. Using a greedy
heuristic, we have also proposed another suboptimal algorithm with
lower complexity, Worst-Channel Removal (WCR), and have shown that
this algorithm is optimal under some channel conditions. Numerical
simulations indicate that these algorithms exhibit only a small gap
to the lower bound. Then, we show that the results obtained for the
single-cycle case can be extended to the multi-cycle scenario with
few modifications. Finally, when the only causal information on the
energy arrivals and only channel statistics are available, we have
introduced a new outage constraint and obtained the optimal resource
allocation.

\section*{Appendix}
\subsection{Proof of Lemma \ref{Lemma-0}}
To prove Lemma~\ref{Lemma-0}, we verify the KKT conditions for the
proposed policy. For $i \in \{1,2, \ldots, N\}$, we have
\begin{align}
 0\le x_i^* &\le c_i, \\
\sum_{i=1}^k x_i^*- \sum_{i=1}^k T_i &\le 0, \\
 \lambda_i (x_i^* - c_i) &=0 \\
 \mu_i (\sum_{j=1}^i x_j^*- \sum_{j=1}^i T_j) &=0 \\
 \gamma_i x_i^* &= 0 \\
 -(\alpha - \beta) +  \lambda_i + \sum_{j=i}^N \mu_j - \gamma_i &= 0 \\
 \lambda_i \ge 0,  \mu_i \ge 0, \gamma_i &\ge 0.
\end{align}
The first two conditions are satisfied since $x_i^* = \min\{c_i,
\sum_{j=1}^{i}T_j - \sum_{j=1}^{i-1} x^*_{j}\}$ is always chosen to
be feasible. It remains to choose $\lambda_i$, $\mu_i$ and
$\gamma_i$. To this end, we set $\gamma_i = 0$. Let $K \in \{1, 2,
\ldots, N\}$ be the largest index such that $x_K^* = \sum_{j=1}^K
T_j - \sum_{j=1}^{K-1} x_j^*$. If there is no such index, we set $K
= 0$. If $K >0$, we set $\mu_K = (\alpha - \beta)$ and $\mu_i = 0$
for $i \neq K$; we set $\lambda_i = (\alpha -\beta)$ for $i > K$ and
$\lambda_i = 0$ for $i \le K$. Note that $(\alpha -\beta) \ge 0$
since $\alpha \ge \beta$. Hence, these choices are feasible. If $K =
0$, we set $\mu_i = 0$ for all $i$ and $\lambda_i  = (\alpha -
\beta)$ for all $i$. It is now easy to verify that all the KKT
conditions are satisfied.

\subsection{Proof of Proposition \ref{Pro-3}}\label{ProofofPro3}
First, we consider the policy with initial storage state $S+\Delta$,
i.e., $\pi(S+\Delta)$. Denote the conventional energy drawn at slot
$i$ by $p_i^c(S+ \Delta)$, and the energy drawn from the renewable
by $p_i^r(S+ \Delta)$. Then, it follows that
\begin{align}
V(S+ \Delta) = \sum_{i=1}^n \left(\alpha p_i^c(S+ \Delta) + \beta
p_i^r(S+ \Delta)\right).
\end{align}
Now, we consider the policy with initial storage state $S$, i.e.,
$\pi(S)$. \textcolor[rgb]{0.00,0.00,0.00}{Let the storage state at
slot $i$ be denoted by $E_i$. For convenience, we introduce the
following indicator functions.}
\textcolor[rgb]{0.00,0.00,0.00}{\begin{align}
\chi_i=\left\{\begin{array}{ll}
                      1, & \mbox{if}~p_i^r(S+ \Delta) \le E_{i-1} + T_i, \\
                      0, & \mbox{otherwise}.
                    \end{array}
\right.\\ \overline{\chi}_i=\left\{\begin{array}{ll}
                      1, & \mbox{if}~p_i^r(S+ \Delta)> E_{i-1} + T_i, \\
                      0, & \mbox{otherwise}.
                    \end{array}
\right.
\end{align}}
Now, we drop the same time slot as in $\pi(S+\Delta)$.  Then, the
drawn energy under $\pi(S)$  can be written as follows,
\begin{align}
&p_i^r(S)  = \chi_i p_i^r(S+ \Delta) + \overline{\chi}_i\left(E_{i-1}+ T_i\right), \label{eq-60}\\
&p_i^c(S) = p_i^c(S+ \Delta) +  \overline{\chi}_i\left(p_i^r(S+
\Delta) - E_{i-1}-T_i\right).
\end{align}
The overall cost of this policy is given by
\begin{align}
V(S) &= \sum_{i=1}^n \left(\alpha p_i^c(S) + \beta p_i^r(S)\right) \nonumber\\
& = \sum_{i=1}^n \alpha (p_i^c(S+ \Delta) +  \overline{\chi}_i (p_i^r(S+ \Delta) - E_{i-1}-T_i)) \nonumber\\
& +\beta(\chi_i p_i^r(S+ \Delta) + \overline{\chi}_i(E_{i-1} +T_i)).
\end{align}
Let $E^\prime_{i}$ be the storage state for slot $i$ under policy
$\pi(S+\Delta)$. Now, we compute  the cost difference between the
two policies.{\allowdisplaybreaks
\textcolor[rgb]{0.00,0.00,0.00}{\begin{align}
&V(S)- V(S+ \Delta) \nonumber\\&= \sum_{i=1}^n \alpha( \overline{\chi}_i(p_i^r(S+ \Delta) - E_{i-1}-T_i))\nonumber\\
& \quad  + \beta \left((\chi_i -1) p_i^r(S+ \Delta)+\overline{\chi}_i(E_{i-1}+ T_i)\right)\nonumber\\
&\stackrel{a}{=}\sum_{i=1}^n (\alpha- \beta)( \overline{\chi}_i(p_i^r(S+ \Delta) - E_{i-1}-T_i)) \nonumber\\
& = \sum_{i=1}^n (\alpha- \beta)\Delta_i,
\end{align}}}
\textcolor[rgb]{0.00,0.00,0.00}{where the equality ``a'' results
from the fact that $\overline{\chi}_i=1-\chi_i$, and
$\Delta_i=\overline{\chi}_i(p_i^r(S+ \Delta) - E_{i-1}-T_i)$.}

\textcolor[rgb]{0.00,0.00,0.00}{Clearly, if $\overline{\chi}_i=0$,
$\Delta_i=0$; if $\overline{\chi}_i=1$, $\Delta_i=p_i^r(S+ \Delta) -
E_{i-1}-T_i\stackrel{b}{=}p_i^r(S+ \Delta) -p_i^r(S)$, where the
equality ``b'' results from the fact that $p_i^r(S)=E_{i-1}+T_i$
when $\overline{\chi}_i=1$ (observed from \eqref{eq-60}).} Hence, we
have
\begin{align}
V(S)- V(S+ \Delta) & = \sum_{i=1}^n (\alpha- \beta)\Delta_i \le
(\alpha -\beta)\Delta.
\end{align}
Proposition \ref{Pro-3} is thus proved.

\subsection{Proof of Proposition \ref{Pro-4}}\label{ProofofPro4}

 Let $\mathcal {S}_1$ denote the set of slots that are kept
in scheme 1, and  $\mathcal {S}_2$ denote the set of slots that are
kept in scheme 2. Denote the leftover harvested energy of scheme 1
as $L$, and that of scheme 2 as $\hat{L}$. Then, it follows that
$L=\sum_{i=1}^N T_i-\sum_{i\in \mathcal {S}_1} p_i^r$, and
$\hat{L}=\sum_{i=1}^N T_i-\sum_{i\in \mathcal {S}_2} \hat{p_i}^r.$
Then, we have \textcolor[rgb]{0.00,0.00,0.00}{\begin{align}
\Delta=\hat{L} - L&=\left(\sum_{i=1}^N T_i-\sum_{i\in \mathcal
{S}_2} \hat{p_i}^r\right)-\left(\sum_{i=1}^N T_i-\sum_{i\in \mathcal
{S}_1} p_i^r \right)\nonumber\\&=\sum_{i\in \mathcal {S}_1} p_i^r
-\sum_{i\in \mathcal {S}_2} \hat{p_i}^r.
\end{align}}
Denote the cost of scheme 1 and scheme 2 as $V$ and $\hat{V}$,
respectively. Then, the cost difference of these two schemes are
given as follows.
\begin{align}
&\hat{V}-V\nonumber\\=&\sum_{i\in \mathcal {S}_2} \left(\alpha
\hat{p}_i^c +\beta \hat{p}_i^r\right)-\sum_{i\in \mathcal {S}_1}
\left(\alpha p_i^c +\beta p_i^r\right)\nonumber\\
=&\alpha\left( \sum_{i\in \mathcal {S}_2} \hat{p}_i^c-\sum_{i\in
\mathcal {S}_1} p_i^c\right)\textcolor[rgb]{0.00,0.00,0.00}{+\beta
\left(\sum_{i\in \mathcal
{S}_2}\hat{p}_i^r-\sum_{i\in \mathcal {S}_1} p_i^r\right)}\nonumber\\
=&\alpha\left[ \kern-0.5mm\sum_{i\in \mathcal {S}_2}\kern-1mm
\left(\frac{N_0(e^R\kern-0.5mm-\kern-0.5mm1)}{g_i}\kern-0.5mm-\kern-0.5mmp_i^r\right)\kern-0.5mm-\kern-1mm\sum_{i\in
\mathcal {S}_1}\kern-1mm
\left(\frac{N_0(e^R\kern-0.5mm-\kern-0.5mm1)}{g_i}\kern-0.5mm-\kern-0.5mmp_i^r\right)\kern-0.5mm\right] \textcolor[rgb]{0.00,0.00,0.00}{\kern-1mm-\kern-0.5mm\beta\Delta}\nonumber\\
=&\alpha\left( \sum_{i\in \mathcal {S}_2}
\frac{N_0(e^R-1)}{g_i}-\sum_{i\in \mathcal {S}_1}
\frac{N_0(e^R-1)}{g_i}\right) +\left(\alpha-\beta\right) \Delta.
\end{align}
It is clear that if scheme $1$ is WCR,  we always have $\sum_{i\in
\mathcal {S}_1} \frac{N_0(e^R-1)}{g_i} < \sum_{i\in \mathcal {S}_2}
\frac{N_0(e^R-1)}{g_i}$. Consequently, we have that $\hat{V}-V >
\left(\alpha-\beta\right) \Delta$. Thus, if WCR is the optimal
solution for each individual cycle $i, ~\forall i$, then it is clear
that the multi-cycle WCR must be the optimal solution for Problem
\ref{Problem-SM}.

\end{document}